\documentclass[11pt]{article}

\usepackage{graphics,color}
\usepackage{graphicx}
\usepackage{amsfonts}
\usepackage{amssymb}
\usepackage{amsmath}
\usepackage{amsthm}
\usepackage{latexsym}
\usepackage{subfigure}

\usepackage[hmargin=2cm,vmargin=2cm]{geometry}

\newtheorem{thm}{Theorem}[section]

\newtheorem{prop}[thm]{Proposition}

\newtheorem{defi}[thm]{Definition}

\newcommand\be{\begin{align}}
\newcommand\ee{\end{align}}
\newcommand\bea{\begin{eqnarray}}
\newcommand\eea{\end{eqnarray}}
\newcommand\bi{\begin{itemize}}
\newcommand\ei{\end{itemize}}
\newcommand\ben{\begin{enumerate}}
\newcommand\een{\end{enumerate}}
\newcommand\bc{\begin{center}}
\newcommand\ec{\end{center}}
\newcommand\ba{\begin{array}}
\newcommand\ea{\end{array}}

\newcommand{\R}{\mathbb{R}}
\newcommand{\N}{\mathbb{N}}

\renewcommand{\P}{\mathbb{P}}

 \newcommand{\scalprod}[2]{\left\langle #1,#2 \right\rangle}

\begin{document}
\renewcommand{\thefootnote}{\fnsymbol{footnote}}
 \footnotetext[1]{
   Hausdorff Center for Mathematics, Universit{\"a}t Bonn, Bonn, Germany
   }
\footnotetext[2]{
    Courant Institute of Mathematical Science, New York University, New York, NY, USA
    }
\renewcommand{\thefootnote}{\arabic{footnote}}

\title{New and improved Johnson-Lindenstrauss embeddings via the Restricted Isometry Property}

\author{Felix Krahmer and Rachel Ward}

\maketitle

\abstract{Consider an $m \times N$ matrix $\Phi$ with the \emph{Restricted Isometry Property} of order $k$ and level $\delta$, that is, the norm of any $k$-sparse vector in $\mathbb{R}^N$ is preserved to within a multiplicative factor of $1 \pm \delta$ under application of $\Phi$.   We show that by randomizing the column signs of such a matrix $\Phi$, the resulting map with high probability embeds \emph{any} fixed set of $p = O(e^{k})$ points in $\mathbb{R}^N$ into $\mathbb{R}^m$ without distorting the norm of any point in the set by more than a factor of $1 \pm 4 \delta$.   Consequently, matrices with the Restricted Isometry Property and with randomized column signs provide optimal Johnson-Lindenstrauss embeddings up to logarithmic factors in $N$.  In particular, our results improve the best known bounds on the necessary embedding dimension $m$ for a wide class of structured random
matrices; for partial Fourier and partial Hadamard matrices, we improve the recent bound $m \gtrsim \delta^{-4} \log(p) \log^4(N)$ appearing in Ailon and Liberty \cite{ailib} to $m \gtrsim \delta^{-2} \log(p) \log^4(N)$, which is optimal up to the logarithmic factors in $N$.  Our results also have a direct application in the area of compressed sensing for redundant dictionaries.  }

\section{Introduction}
The \emph{Johnson-Lindenstrauss} (JL) Lemma states that any set of $p$ points in high dimensional Euclidean space can be embedded into $O(\varepsilon^{-2} \log(p))$ dimensions, without distorting the distance between any two points by more than a factor between $1 - \varepsilon$ and $1 + \varepsilon$.  In its original form, the Johnson-Lindenstrauss Lemma reads as follows.
\begin{thm}[Johnson-Lindenstrauss Lemma  \cite{JL}]
\label{JL}
Let $\varepsilon \in (0,1)$ and let $x_1, ..., x_p \in \mathbb{R}^N$ be arbitrary points.  Let $m = O(\varepsilon^{-2} \log(p))$ be a natural number.  Then there exists a Lipschitz map $f: \mathbb{R}^N \rightarrow \mathbb{R}^m$ such that
\begin{align}
\label{embed}
(1 - \varepsilon) \| x_i-x_j \|_2^2 \leq \| f (x_i)-f(x_j) \|_2^2 \leq (1 + \varepsilon) \| x_i-x_j \|_2^2
\end{align}
for all $i,j \in \{1,2, ..., p\}$.  Here $\| \cdot \|_2$ stands for the Euclidean norm in $\mathbb{R}^N$ or $\mathbb{R}^m$, respectively.
\end{thm}
As shown in \cite{alon}, the bound for the size of $m$ is tight up to an $O(\log(1/\varepsilon))$ factor.  In the original paper of Johnson and Lindenstrauss, it was shown that a random orthogonal projection, suitably normalized, provides such an embedding with high probability \cite{JL}.  Later, this property was also verified for Gaussian random matrices, among other random matrix constructions \cite{FM, Dasgupta}.  As a consequence,  the JL Lemma has become a valuable tool for dimensionality reduction in a myriad of applications ranging from computer science \cite{ind-01}, numerical linear algebra \cite{sar,hmt, liwomaroty07}, manifold learning \cite{randomprojections}, and compressed sensing \cite{badadewa08}, \cite{rw09}, \cite{caneel10}.   

In most of these frameworks, the map $f$ under consideration is a linear map represented by an $m\times N$ matrix $\Phi$. In this case, one can consider the set of differences $E=\{x_i-x_j\}$; to prove the theorem, one then needs to show that
\begin{align}
(1 - \varepsilon) \| y \|_2^2 \leq \| \Phi y \|_2^2 \leq (1 + \varepsilon) \| y \|_2^2, \hspace{3mm} \textrm{ for all } y \in E. \label{eq:JLlin}
\end{align}
When $\Phi$ is a random matrix,
the proof that $\Phi$ satisfies the JL lemma with high probability boils down to showing a concentration inequality of the type
\begin{align}
\P\big((1-\varepsilon) \|x\|_{2}^2 \leq \| \Phi x \|_{2}^2 \leq (1+\varepsilon) \|x\|_{2}^2 \big) \geq 1 - 2 \exp(-c_0 \varepsilon^2 m ), \label{eq:concineq}
\end{align}
for an arbitrary fixed $x\in \R^N$, where $c_0$ is an absolute constant in the optimal case, and in addition possibly mildly dependent on $N$ in almost-optimal scenarios as for example in \cite{ailib}. Indeed it  directly follows by a union bound over $E$ (as in the proof of Theorem~\ref{thm:main} below) that  \eqref{eq:JLlin} holds with high probability.

In order to reduce storage space and implementation time of such embeddings, the design of structured random JL embeddings has been an active area of research in recent years  \cite{ailibsing, doganchentran, ailib, kanel}; see \cite{ailibsing} or \cite{kanel} for a good overview of these efforts. Of particular importance in this context is whether fast (i.e. $O(N\log(N))$) multiplication algorithms are available for the resulting matrices. Fast JL embeddings with optimal embedding dimension $m = O(\varepsilon^{-2} \log(p))$ were first constructed by Ailon and Chazelle \cite{fastjlt}, but their embeddings are fast only for $p \lesssim e^{N^{1/3}}$ vectors. This restriction on the number of vectors was later weakened to $p \lesssim e^{N^{1/2}}$ \cite{ailonlib08}. In \cite{ailib}, fast JL embeddings were constructed without any restrictions on the number of vectors, but the authors only provide sub-optimal embedding dimension $m = O(\varepsilon^{-4} \log(p)\log^4(N))$.  In this paper, we provide the first unrestricted fast JL construction with optimal embedding dimension up to logarithmic factors in $N$.  Note that in the range $p \gtrsim e^{N^{1/2}}$ not covered by the constructions in \cite{fastjlt, ailonlib08}, a logarithmic factor in $N$ is bounded by $\log(\log(p))$, and thus plays a minor role.

\paragraph{The Johnson-Lindenstrauss Lemma in Compressed Sensing.}
One of the more recent applications of the Johnson-Lindenstrauss Lemma is to the area of \emph{compressed sensing}, which is centered around the following phenomenon:  For many underdetermined systems of linear equations $\Phi x = y$, the solution of minimal $\ell_1$-norm is also the sparsest solution.  To be precise, a vector $x \in \mathbb{R}^N$ is $k$-sparse if $ | \{ j : | x_j | > 0 \} | \leq k$.  A by now classical sufficient condition on the matrix $\Phi$ for guaranteeing equivalence between the minimal $\ell_1$ norm solution and sparsest solution is the so-called \emph{Restricted Isometry Property} (RIP)\cite{carota06, cata06, do06}.
\begin{defi}
A matrix $\Phi \in \mathbb{R}^{m \times N}$ is said to have the Restricted Isometry Property of order $k$ and level $\delta \in (0,1)$ (equivalently, $(k, \delta)$-RIP) if
\begin{align}
\label{rip}
(1 - \delta) \| x \|_2^2 \leq \| \Phi x \|_2^2 \leq (1 + \delta) \| x \|_2^2 \hspace{12mm} \textrm{ for all}\hspace{2mm} k \textrm{-sparse } x \in \mathbb{R}^N .
\end{align}
The restricted isometry constant $\delta_k$ is defined as the smallest value of $\delta$ for which \eqref{rip} holds.
\end{defi}
In particular, if $\Phi$ has $(2k, \delta_{2k})$-RIP with $\delta_{2k} \leq  2/(3 + \sqrt{\frac{7}{4}}) \approx .4627$, and if $y = \Phi x$ admits a $k$-sparse solution $x^{\#}$, then $x^{\#} = \arg \min_{\Phi z = y} \| z \|_1$ \cite{fo09}.

Gaussian and Bernoulli random matrices have $(k, \delta)$-RIP with high probability, if the embedding dimension $m \gtrsim \delta^{-2} k \log(N/k)$ \cite{badadewa08}.   Up to the constant, lower bounds for Gelfand widths of $\ell_1$-balls \cite{gagl84, foparaul10} show that this dependence on $N$ and in $k$ is optimal.  The Restricted Isometry Property also holds for a rich class of structured random matrices, where usually the best known bounds for $m$ have additional log factors in $N$.  All known deterministic constructions of RIP matrices require that $m \gtrsim k^2$ or  at least $m \gtrsim k^{2-\mu}$ for some small constant $\mu > 0$ \cite{bdfkk}.

The similarity between the expressions in \eqref{eq:JLlin} and \eqref{rip} suggests a connection between the JL lemma and the Restricted Isometry Property.  A first result in this direction was established in \cite{badadewa08}, wherein it was shown that random matrices  satisfying a concentration inequality of type \eqref{eq:concineq} (and hence the JL Lemma) satisfy the RIP of optimal order.  More precisely, the authors prove the following theorem.

\begin{thm}[Theorem 5.2 in \cite{badadewa08}] \label{thm:JLRIP}
Suppose that $m, N$, and $0<\delta<1$ are given. If the probability distribution generating the $m\times N$ matrices $\Phi$ satisfies the concentration inequality \eqref{eq:concineq} with $\varepsilon = \delta$ and absolute constant $c_0$, then there exist absolute constants $c_1,c_2$ such that with probability $\geq 1-2 e^{-c_2 \delta^{2} m},$ the RIP \eqref{rip} holds for $\Phi$ with the prescribed $\delta$ and any $k\leq c_1 \delta^{2} m/\log(N/k)$.
\end{thm}

\noindent In this sense, the JL Lemma \emph{implies} the Restricted Isometry Property.

\paragraph{Contribution of this work.} We prove a converse result to Theorem \ref{thm:JLRIP}:  We show that RIP matrices, with randomized column signs, provide Johnson-Lindenstrauss embeddings that are optimal up to logarithmic factors in the ambient dimension.
In particular, RIP matrices of optimal order provide Johnson-Lindenstrauss embeddings of optimal order as such, up to a logarithmic factor in $N$ (see Theorem~\ref{thm:main}). Note that without randomization, such a converse is impossible as vectors in the null space of the fixed parent matrix are always mapped to zero.

This observation has several consequences in the area of compressed sensing, and also allows us to obtain improved JL embedding results for several matrix constructions with existing RIP bounds  \cite{cata06, ruve08, ra09-1, badularotr10, rw10}.  Of particular interest is the random partial Fourier or the random partial Hadamard matrix, which is formed by choosing a random subset of $m$ rows from the $N \times N$ discrete Fourier or Hadamard matrix respectively,  and with high probability has $(k, \delta)$-RIP if the embedding dimension $m \gtrsim \delta^{-2} k \log^4(N)$.   For these matrices with randomized column signs, the running time for matrix-vector multiplication is $O(N\log(N))$ as opposed to the running time of $O(Nm)$ for purely random matrices.  For such constructions, the previous best-known embedding dimension to ensure that \eqref{eq:JLlin} holds with probability $1 - \eta$, given by Ailon and Liberty \cite{ailib}, is $m \asymp \varepsilon^{-4}\log(p/\eta)\log^4(N)$.   We can improve their result to have optimal dependence on the distortion, $\varepsilon$, showing that $m \asymp \varepsilon^{-2}\log(p/\eta)\log^4(N)$ rows suffice for the embedding.

This paper is structured as follows:  Section \ref{sec:notation} introduces necessary notation. In Section \ref{sec:main}, we state our main results, and Section~\ref{sec:ex} gives concrete examples of how these results improve on the best-known JL bounds for several matrix constructions as well as applications of our findings in compressed sensing.  In Section \ref{sec:ingr} we give  the relevant concentration inequalities and explicit RIP-based matrix inequalities that are needed for the proofs, which are then carried out in Section \ref{sec:proof}.

\section{Notation\label{sec:notation}}
Before continuing, let us fix some notation to be used in the remainder. For $N\in\N$, we denote $[N] = \{1,\hdots,N\}$.
The $\ell_p$-norm of a vector $x = (x_1,\dots, x_N) \in \mathbb{R}^N$ is defined as
\[
\|x\|_p = \big( \sum_{j=1}^N |x_j|^p \big)^{1/p}, \quad 1 \leq p < \infty,
\]
and $\|x\|_\infty = \max_{j=1,\hdots,N} |x_j|$ as usual.  For a matrix $\Phi = (\Phi_{j,\ell}) \in \mathbb{R}^{m \times N}$, its operator norm is  $\| \Phi \| := \sup_{\| x \|_2 = 1} \| \Phi x \|_2$, and its Frobenius norm is defined by
\[
\| \Phi \|_{\cal{F}} :=  \big( \sum_{j=1}^{m} \sum_{\ell=1}^N | \Phi_{j,\ell} |^2 \big)^{1/2}.
\]
For two functions $f, g:S\rightarrow \R^+$, $S$ an arbitrary set, we write $f\gtrsim g$ if there is a constant $C>0$ such that $f(x)\ge C g(x)$ for all $x\in S$; we write $f\asymp g$ if $f \gtrsim g$ and $g\gtrsim f$.
Let $N$ and $s\ll N$ be given and set $R=\left\lceil\frac{N}{s}\right\rceil$.  For given $x = (x_1,\dots, x_N) \in \mathbb{R}^N$, we say that $x$ is in {\em decreasing arrangement}, if one has $|x_i|\geq |x_j|$ for $i<j$.  For vectors in decreasing arrangement, we decompose $x = (x_{(1)}, \dots , x_{(J)}, \dots , x_{(R)})$ into blocks of size $s = k/2$, i.e. $x_{(J)} \in \mathbb{R}^s$; the last block $x_{(R)}$ is potentially of smaller size. We will also consider the coarse decomposition $x = (x_{(1)}, x_{(\flat)})$, where $x_{(\flat)} = (x_{(2)}, ..., x_{(R)}) \in \mathbb{R}^{N-s}$.  Denote by $]L[$ the indices corresponding to the $L$-th block. For $ j, \ell \in [N]$ we write $j \sim l$ if the two indices are associated to the same block, and we write $j\nsim \ell$ otherwise.
Given a matrix $\Phi\in\mathbb{R}^{m \times N}$,  write $\Phi_j$ to denote the $j$-th column, $\Phi_{(J)} \in \mathbb{R}^{m \times s}$ to denote the matrix that is the restriction of $\Phi$ to the $s$ columns indexed by $J$ (again with the obvious modification for $J=R$), and $\Phi_{(\flat)}$ to denote the restriction of $\Phi$ to all but the first $k$ columns.
Finally, for a vector $x \in \mathbb{R}^N$, we denote by $D_x = (D_{i,j}) \in \mathbb{R}^{N \times N}$ the diagonal matrix satisfying $D_{j,j} = x_j$.

\section{The main results \label{sec:main}}

\begin{thm}\label{thm:main}
Fix $\eta > 0$ and $\varepsilon \in (0,1)$, and consider a finite set $E \subset\mathbb{R}^N$ of cardinality $|E| = p$.  Set $k \geq 40\log{\frac{4p}{\eta}}$,
and suppose that $\Phi \in \mathbb{R}^{m \times N}$ satisfies the Restricted Isometry Property of order $k$ and level $\delta \leq \frac{\varepsilon}{4}$.  Let $\xi \in \mathbb{R}^N$ be a Rademacher sequence, i.e., uniformly distributed on $\{-1,1\}^N$.
Then with probability exceeding $1 - \eta$,
\begin{align}
\label{embedding}
(1 - \varepsilon) \| x \|_2^2 \leq \| \Phi D_{\xi} x \|_2^2 \leq (1 + \varepsilon) \| x \|_2^2
\end{align}
uniformly for all $x\in E$.
\end{thm}

Along the way, our method provides a direct converse to Theorem~\ref{thm:JLRIP}:
\begin{prop}\label{prop:conv}
Fix $\varepsilon \in (0,1)$, and suppose that there is a constant $c_3$ such that for all pairs $(k,m)$ that are admissible in the sense that $k\leq c_3 \delta^2 m/\log(N/k)$, $\Phi=\Phi(m) \in \mathbb{R}^{m \times N}$ has the Restricted Isometry Property of order $k$ and level $\delta \leq \frac{\varepsilon}{4}$.  Fix $x\in\R^N$ and let $\xi \in \mathbb{R}^N$ be a Rademacher sequence, i.e., uniformly distributed on $\{-1,1\}^N$. Then there exists a constant $c_4$ such that  for all $m$, $\Phi D_\xi$ satisfies the concentration inequality \eqref{eq:concineq} for $c_0= c_4 \log^{-1} \left( \frac{N}{k}\right)$, where $k$ is any integer such that $(k,m)$ is admissible.
\end{prop}

\section{Concrete examples and applications\label{sec:ex}}
Using Theorem~\ref{thm:main}, we can improve on the best Johnson-Lindenstrauss bounds for several matrix constructions that are known to have the Restricted Isometry Property:

\paragraph{\bf 1.  Matrices arising from bounded orthonormal systems.}
Consider an orthonormal system of real-valued functions $\varphi_j$, $j \in [N]$, on a measurable space ${\cal S}$ with respect to an orthogonalization measure $d\nu$.  Such systems are called \emph{bounded orthonormal systems} if $ \sup_{j \in [N]} \sup_{x \in {\cal S}} |\varphi_j(x)| \leq K$
for some constant $K\geq 1$.   We may associate to such a system the $m \times N$ matrix $\Phi$ with entries $\Phi_{\ell,j} = \frac{1}{\sqrt{m}} \varphi_{j}(x_{\ell})$, where $x_{\ell}, \hspace{1mm} \ell \in [m]$, are drawn independently according to the orthogonalization measure $d\nu$.  As shown in \cite{cata06, ruve08, ra09-1}, matrices arising as such have $(k, \delta)$-RIP with high probability if $m \gtrsim \delta^{-2} k  \log^4(N)$.
By Theorem \ref{thm:main}, these embeddings with randomized column signs satisfy the JL Lemma for $m \gtrsim \varepsilon^{-2} \log(p) \log(N)$, which is optimal up to the $\log(N)$ factors.\footnote{Actually, the bounds in \cite{ra09-1} yield that $m \gtrsim \delta^{-2} k \log^3(k) \log^2(N)$ is sufficient for $\Phi$ to have $(k,\delta)$-RIP with high probability. Hence $\Phi D_\xi$ is a JL-embedding for $m \gtrsim \varepsilon^{-2} \log(p) \log^3{(\log(p))} \log(N)$. However, in order to work with simpler expressions, we bound $k<N$ in the logarithmic factors.}

 For measures with discrete support, such constructions are equivalent to choosing $m$ rows at random from an $N \times N$ matrix with orthonormal rows and uniformly bounded entries.   Examples include the random partial Fourier matrix or random partial Hadamard matrix, formed from the discrete Fourier matrix or discrete Hadamard matrix respectively. (In the Fourier case, we distribute the resulting real and complex parts in different coordinates, inducing an additional factor of $2$.)  Note that the structure of these matrices allows for fast matrix vector multiplication.  Recently, Ailon and Liberty \cite{ailib} verified the JL Lemma for such constructions, with column signs randomized, when $m \gtrsim \varepsilon^{-4} \log(p) \log^4 (N)$.  Our result improves the factor of $\varepsilon^{-4}$ in their result to the optimal dependence $\varepsilon^{-2}$.  We note that while their proof also uses the RIP, it also requires arguments from \cite{ruve08} that are specific to discrete bounded orthonormal systems.

Examples of bounded orthonormal systems connected to continuous measures include the trigonometric polynomials and Chebyshev polynomials, which are orthogonal with respect to the uniform and Chebyshev measures, respectively.  The Legendre system, while not uniformly bounded, can still be transformed via preconditioning to a bounded orthonormal system with respect to the Chebyshev measure \cite{rw10}.  Note that all of these constructions have an associated fast transform.

\paragraph{\bf 2.  Partial circulant matrices.} Other classes of structured random matrices known to have the RIP include \emph{partial circulant} matrices \cite{rom09, ra09, rrt10}.  In one such set-up, the first row of the $N \times N$ matrix is a Gaussian or Rademacher random vector, and each subsequent row is created by rotating one element to the right relative to the preceding row vector. Again, $m$ rows of this matrix are sampled, but in contrast to partial Fourier or Hadamard matrices, the selection need not be random. Using that convolution corresponds to multiplication in the Fourier domain, these matrices have associated fast matrix-vector multiplication routines.  In \cite{rrt10}, such matrices were shown to have the RIP with high probability for $m\gtrsim \operatorname{max}\left(\delta^{-1} k^{\frac{3}{2}} \log^{\frac{3}{2}}(N) ,\delta^{-2} k \log^{4}(N) \right)$.

On the other hand, such a matrix  composed with a diagonal matrix of random signs was shown to be a JL embedding with high probability as long as $m\gtrsim \varepsilon^{-2} \log^2(p)$ \cite{jlcirc}.
Through Theorem \ref{thm:main}, the same results also obtain if $m\gtrsim \operatorname{max}\Big(\varepsilon^{-1}\log^{3/2}\left(\frac{4p}{\eta} \right) \log^{\frac{3}{2}}(N), \varepsilon^{-2}\log\left(\frac{4p}{\eta}\right) \log^{4}(N)\Big)$. For large $p$, this is an improvement compared to \cite{jlcirc}.

\paragraph{\bf 3.  Deterministic constructions.} Several deterministic constructions of RIP matrices are known, including a recent result in \cite{bdfkk} that requires only $m \gtrsim k^{2 - \mu}$.  We refer the reader to the exposition in \cite{bdfkk} for a good overview in this direction;  we highlight two such deterministic constructions here.
 Using finite fields, DeVore \cite{de07-3} provides deterministic constructs of cyclic $0$-$1$-valued matrices with $(k, \delta)$-RIP with $m \gtrsim \delta^{-2} k^2 \log^2(N)$.  Iwen \cite{iwen10} provides deterministic constructions of $0$-$1$-valued matrices whose number theoretic properties allow their products with Discrete Fourier Transform (DFT) matrices to be well approximated using a few highly sparse matrix multiplications.  Both the binary-valued matrices and their products with the DFT yield $(k, \delta)$-RIP matrices with $m \gtrsim \delta^{-2} k^2 \log^2(N)$.  By Theorem \ref{thm:main}, the class of matrices that results by randomizing the column signs of either of these deterministic constructions satisfies the JL Lemma with $m \gtrsim \varepsilon^{-2} \log^2(p) \log^2(N)$.

Note that the amount of randomness needed to construct such embeddings is still comparable to the first two examples, requiring $N$ random bits. Under the model assumption that the entries of each vector $x\in E$ to be embedded has random signs, however, the required randomness in the matrix is removed completely.

%

\medskip
\medskip

In addition to their fast multiplication properties, these examples have the advantage in that the construction of the matrix embedding only uses $N+m$, $2N+m$, and $N$ independent random bits, respectively, compared to $m N$ bits for matrices with independent entries. We note  that stronger embedding results are known with fewer bits, if one imposes restrictions on the $\ell_{\infty}$ norm of the vectors $x \in E$ to be embedded -- see \cite{kanel} and \cite{sparse10}.

\medskip

For each of the aforementioned examples, we summarize the number of dimensions $m$ that are known to be sufficient $(k, \delta)$-RIP to hold.  We also list the previously best known bound for  JL embedding dimension  (if there is one) along with the JL bounds obtained from Theorem \ref{thm:main}.  Where Theorem \ref{thm:main} yields a better bound than previously known, at least for some range of parameters, we highlight the result in bold face.  In each of the bounds, we list only the dependence on $\delta, k$, and $N$, or $\varepsilon, k,$ and $N$, omitting absolute constants.

\medskip

\begin{table}[h!b!p!]
\begin{center}
\begin{tabular}{| c | c | c | c | }
\hline
& RIP bounds & Previous JL Bound & JL Bound from Theorem $\ref{thm:main}$ \\
&&&\\
\hline
&&& \\
Partial Fourier & $\delta^{-2} k \log^4(N)$ &
$\varepsilon^{-4} \log(\frac{p}{\eta}) \log^4(N)$  & { \boldmath $ \varepsilon^{-2} \log(\frac{p}{\eta}) \log^4(N)$}
 \\
 &&&\\
 \hline
 & & & \\
Partial Circulant  & $\operatorname{max}\left(\delta^{-1}{k^{\frac{3}{2}} \log^{\frac{3}{2}}(N)},\right.$  & $\varepsilon^{-2} \log^2{(\frac{p}{\eta})}$ & {\boldmath $\operatorname{max}\left(\varepsilon^{-1}{\log^{\frac{3}{2}}(\frac{p}{\eta}) \log^{\frac{3}{2}}(N)},\right.$}\\
&$\left.{\delta^{-2}}k \log^{4}(N) \right)$&&{ \boldmath $\left.{\varepsilon^{-2}}\log(\frac{p}{\eta})\log^{4}(N) \right)$ } \\
\hline
&&&\\
Deterministic & $\delta^{-2} k^2 \log^2(N)$ &  & \boldmath${ \varepsilon^{-2} \log^2{(\frac{p}{\eta})} \log^2(N)}$ \\
(DeVore, Iwen) &&&\\
 \hline
 &&&\\
Subgaussian  & $\delta^{-2} k \log{(\frac{N}{k})}$ & $\varepsilon^{-2} \log{(\frac{p}{\eta})}$ & $\varepsilon^{-2} \log{(\frac{p}{\eta})} \log(N)$ \\
&&&\\
\hline
\end{tabular}
\end{center}
\end{table}

\paragraph{\bf 4. Compressed sensing in redundant dictionaries.}
As shown recently in \cite{caneel10}, concentration inequalities of type \eqref{eq:concineq} allow for the extension of the compressed sensing methodology to redundant dictionaries -- in particular, tight frames -- as opposed to orthonormal bases only. Since signals with sparse representations in redundant dictionaries comprise a much
more realistic model of nature, this extension of compressed sensing is fundamental. Our results show that basically all random matrix constructions arising in the standard theory of compressed sensing (i.e., based on RIP  estimates) also yield compressed sensing matrices for the redundant framework.

\paragraph{\bf 5. Compressed sensing with cross validation.}
Compressed sensing algorithms are designed to recover approximately sparse signals; if this assumption is violated, they may yield solutions far from the input signal.  In \cite{rw09}, a method of \emph{cross validation} is introduced to detect such situations, and to obtain tight bounds on the error incurred by compressed sensing reconstruction algorithms in general.  There, a subset $y_1 = \Phi_1 x$ of the $m$ measurements $y=\Phi x$ are held out from the reconstruction algorithm and only the remaining measurements $y_2=\Phi_2 x$ are used to produce a candidate approximation $\widehat{x}$ to the unknown $x$. If the hold-out matrix $\Phi_1$ satisfies the Johnson-Lindenstrauss Lemma, then the observable quantity $\| \Phi_1 (x - \widehat{x}) \|_2$ can be used as a reliable proxy for the unknown error $\| x - \widehat{x} \|_2$.  Our work shows that any RIP matrix as in the standard compressed sensing framework can be used for cross validation up to a randomization of its column signs.

\paragraph{\bf 6. Optimal asymptotics in $\delta$ for RIP to hold.}
As mentioned above, it can be shown using a Gelfand width argument that $m \asymp k \log(\frac{N}{k})$ is the optimal asymptotics (in $N$ and $k$) of the embedding dimension for a matrix with the restricted isometry property \eqref{rip}.  Our results -- combined with the known optimality of the asymptotics $m = \varepsilon^{-2} \log(p)$ for the embedding dimension in the Johnson-Lindenstrauss Lemma \eqref{JL} -- imply that up to a factor of $\log\left(\frac{1}{\delta}\right)$, $m\asymp \delta^{-2}$ is the optimal asymptotics  in the restricted isometry constant $\delta$ for fixed $N$ and $k$ as $\delta\rightarrow 0$. Recall that this rate is realized by many of the above examples, such as Gaussian random matrices.
\section{Proof Ingredients \label{sec:ingr}}
The proof of Theorem \ref{thm:main} relies on concentration inequalities for Rademacher sequences and explicit RIP-based norm estimates. The first concentration result is a classical  inequality by Hoeffding \cite{ho63}.
\begin{prop}[Hoeffding's Inequality]
\label{hoeff}
Let $x \in \mathbb{R}^N$, and let $\xi = (\xi_j)_{j=1}^N$ be a Rademacher sequence.  Then, for any $t > 0$,
\begin{align}
\label{bern}
\mathbb{P} \Big( | \sum_{j} \xi_j x_j  | > t  \Big) \leq 2 \exp\Big(-\frac{t^2}{2\| x \|_2^2}\Big).
\end{align}
\end{prop}

The second concentration of measure result is a deviation bound for Rademacher chaos. There are many such bounds in the literature; the following inequality dates back to \cite{HaWr71}, but appeared with explicit constants and with a much simplified proof as Theorem  $17$ in \cite{boluma03}.
\begin{prop}
\label{chaos}
Let $X$ be the $N\times N$ matrix with entries $x_{i,j}$ and assume that $x_{i,i} = 0$ for all $i\in[N]$.  Let $\xi = (\xi_j)_{j=1}^N$ be a Rademacher sequence. Then, for any $t > 0$,
\begin{align}
\label{ineq}
\mathbb{P} \Big( | \sum_{i,j} \xi_i \xi_j x_{i,j} | > t  \Big) \leq 2 \exp\Big( -\frac{1}{64} \min\Big( \frac{\frac{96}{65}t}{\| X \|}, \frac{t^2}{\| X \|_{\cal F}^2} \Big) \Big).
\end{align}
\end{prop}

We also need the following basic estimate for RIP matrices (see for instance Proposition $2.5$ in \cite{ra09-1}).
\begin{prop}
\label{RIP-est}
Suppose that  $\Phi \in \mathbb{R}^{m \times N}$ has the Restricted Isometry Property of order $2s$ and level $\delta$.  Then for any two disjoint subsets $J, L \subset [N]$ of size $| J | \leq s, | L | \leq s$,
\begin{align}
\| \Phi_{(J)}^* \Phi_{(L)} \| \leq \delta.
\nonumber
\end{align}
\end{prop}
The proof of our norm estimate for RIP-matrices uses Proposition \ref{RIP-est}, and relies on the observation commonly used in the theory of compressed sensing (see for example \cite{carota06-1}) that for $z$ in decreasing arrangement and $\| z \|_2 =1$, for $J\geq 2$ one has $\|z_{(J)} \|_\infty\leq \frac{1}{\sqrt{s}}\|z_{(J-1)}\|_2$ and thus $\| z_{(\flat)} \|_{\infty} \leq 1/\sqrt{s}$.
\begin{prop}
\label{C}
Let $R = \lceil{ N/s \rceil}$.  Let $\Phi = (\Phi_j) = (\Phi_{(1)}, \Phi_{(2)}, ..., \Phi_{(R)} ) = (\Phi_{(1)}, \Phi_{(\flat)} ) \in\mathbb{R}^{m \times N}$ have the $(2s, \delta)$-Restricted Isometry Property, let $x = (x_j) = (x_{(1)}, x_{(2)}, ... , x_{(R)}) = (x_{(1)}, x_{(\flat)}) \in\mathbb{R}^N$ be in decreasing arrangement with $\| x \|_2 \leq 1$, and consider the symmetric matrix
\begin{align}
C \in \mathbb{R}^{N \times N}, \quad C_{j,\ell} = \left\{ \begin{array}{ll} x_j \Phi_j^* \Phi_\ell x_\ell,& j \nsim \ell, \quad j, \ell > s, \\
0, & \textrm{else},
\nonumber
\end{array} \right.
\end{align}
and, for $b \in\{-1,1\}^s$, the vector
\begin{align}
v \in \mathbb{R}^N, \quad v = D_{x_{(\flat)}} \Phi_{(\flat)}^*\Phi_{(1)} D_{x_{(1)}} b.
\nonumber
\end{align}
The following bounds hold: $\quad \|C\| \leq  \frac{\delta}{s}, \quad \|C\|_{\cal{F}} \leq  \frac{ \delta}{\sqrt{s}}, \quad$ and $\quad \|v\|_2 \leq \frac{\delta}{\sqrt{s}}$.
\end{prop}

\begin{proof}
\begin{align}
\|C\|=&\sup_{\|y\|_2=1} \left|  \scalprod{y}{Cy} \right| \nonumber \\
\leq& \sup_{\|y\|_2=1} \sum_{\substack{J, L =2\\ J\neq L }}^R \left| \scalprod{y_{(J)}}{D_{x_{(J)}} \Phi^*_{(J)}\Phi_{(L)}D_{x_{(L)}} y_{(L)}} \right| \nonumber \\
\leq& \sup_{\|y\|_2=1} \sum_{\substack{J, L =2\\ J\neq L }}^R \|y_{(J)}\|_2\| y_{(L)}\|_2 \|D_{x_{(J)}} {\Phi^*_{(J)}} \Phi_{(L)}D_{x_{(L)}} \| \nonumber \\
\leq & \sup_{\|y\|_2=1}\sum_{\substack{J, L =2\\ J\neq L }}^R \| y_{(J)}\|_2\| y_{(L)}\|_2 \|x_{(J)}\|_\infty \|x_{(L)} \|_\infty \delta\label{eq:RIPest} \\
\leq  & \sup_{\|y\|_2=1} \sum_{J, L =2 }^R \| y_{(J)}\|_2\| y_{(L)}\|_2 \frac{1}{\sqrt{s}}\|x_{(J-1)}\|_2 \frac{1}{\sqrt{s}} \|x_{(L-1)}\|_2 \delta \hspace{10mm}  \nonumber \\
\leq  & \sup_{\|y\|_2=1} \frac{\delta}{s} \sum_{J, L =2}^R \left(\frac{1}{2} \|x_{(J-1)}\|^2_2+\frac{1}{2} \| y_{(J)}\|^2_2\right) \left(\frac{1}{2} \| x_{(L-1)}\|^2_2+\frac{1}{2} \|y_{(L)}\|^2_2\right) \label{eq:ArGeo}\\
\leq & \frac{\delta}{s}. \nonumber
\end{align}
To obtain \eqref{eq:ArGeo}, we use the inequality of arithmetic and geometric means; to obtain \eqref{eq:RIPest}, we use Proposition~\ref{RIP-est}.

Similarly,
\begin{align}
\|v\|_2 \leq & \sup_{\|y\|_2=1} \sum_{L = 2}^R \scalprod{y_{(L)}}{D_{x_{(L)}}^* \Phi_{(L)}^*\Phi_{(1)} D_{(b)} x_{(1)}} \nonumber \\
\leq & \sup_{\|y\|_2=1} \sum_{L = 2}^R \| y_{(L)}\|_2 \| x_{(L)}\|_\infty \| \Phi_{(L)}^*\Phi_{(1)}\| \|b\|_\infty \| x_{(1)} \|_2 \nonumber \\
\leq & \sup_{\|y\|_2=1} \sum_{L = 2}^R \|y_{(L)}\|_2 \frac{1}{\sqrt{s}} \|x_{(L-1)}\|_2 \| \Phi_{(L)}^*\Phi_{(1)}\| \|b\|_\infty \nonumber  \\
\leq & \frac{\delta}{\sqrt{s}} \sup_{\|y\|_2=1} \sum_{L = 2}^R \left(\frac{1}{2}\|y_{(L)}\|^2_2  +\frac{1}{2}\|x_{(L-1)}\|^2_2\right)  \nonumber \\
\leq & \frac{\delta}{\sqrt{s}}. \nonumber
\end{align}
For the Frobenius norm, we estimate:
\begin{align}
\|C\|^2_{\cal F}= & \sum_{\substack{j,l=s+1\\ j\nsim \ell}}^N (x_j \Phi_j^* \Phi_\ell x_\ell)^2 \nonumber \\
= & \sum_{\substack{L=2}}^R  \sum_{\substack{ j=s+1\\ j\notin ]L[}}^N x_j^2   \Phi_j^* \Phi_{(L)} D^2_{x_{(L)}} \Phi_{(L)}^* \Phi_j \nonumber \\
= &\sum_{L=2}^R \sum_{\substack{ j=s+1\\ j\notin ]L[}}^N   x_j^2 \| D_{ x_{(L)}} \Phi_{(L)}^* \Phi_j \|^2 \nonumber \\
\leq &\sum_{L=2}^R \sum_{\substack{ j=s+1\\ j\notin ]L[}}^N    x_j^2  \|x_{(L)}\|_\infty^2 \| \Phi_{(L)}^* \Phi_j\|^2 \nonumber \\
\leq &\sum_{L=2}^R  \frac{\delta^2}{s} \|x_{(L-1)}\|_2^2  \sum_{j=1}^{N}   x_j^2 \nonumber \\
\leq &\frac{\delta^2}{s}. \nonumber
\end{align}
\end{proof}

\section{Proof of the main results\label{sec:proof}}

We begin by proving Theorem \ref{thm:main}.  Without loss of generality, we assume that all $x \in E$ are normalized so that $\| x \|_2 = 1$. Furthermore, assume that $k=2 s$ is even.

We first consider a fixed $x\in E$, eventually taking a union bound over all $x$. We further assume that $x$ is in decreasing arrangement. To achieve this, we reorder the entries of $x$, and permute the columns of  $\Phi$ accordingly. This has no impact on the following estimates, as the Restricted Isometry Property of the matrix $\Phi$ is invariant under permutations of its columns.
 We need to estimate
\begin{align}
\| \Phi& D_\xi x \|_2^2  \hspace{1mm} = \hspace{1mm}  \| \Phi D_x \xi \|_2^2 \nonumber \\ =& \hspace{1mm}  \sum_{J=1}^R \| \Phi_{(J)} D_{x_{(J)}} \xi_{(J)}  \|_2^2 \hspace{1mm} +  \hspace{1mm} 2 \xi_{(1)}^{*} D_{x_{(1)}} \Phi_{(1)}^* \Phi_{(\flat)} D_{x_{(\flat)}} \xi_{(\flat)} \hspace{1mm} + \hspace{1mm} \sum_{\substack{J, L =2\\ J\neq L }}^R\scalprod{\Phi_{(J)} D_{x_{(J)}} \xi_{(J)}}{\Phi_{(L)} D_{x_{(L)}} \xi_{(L)}}. \label{eq:expansion}
\end{align}
We will bound the terms separately.

\begin{enumerate}
\item
As $\Phi$ has the Restricted Isometry Property of order $k\geq s$ and level $\delta$, it also has the RIP of order $s$ and level $\delta$, and each $\Phi_{(J)}$ is almost an isometry.
Hence, noting that $\|  D_{x_{(J)}} \xi_{(J)} \|_2 = \|D_{\xi_{(J)}}  x_{(J)}\|_2 = \|x_{(J)}\|_2$,
the first term can be estimated as follows.
\begin{align}
\nonumber
(1 - \delta) \| x \|_2^2 \leq \sum_{J=1}^R \| \Phi_{(J)} D_{x_{(J)}} \xi_{(J)}  \|_2^2 \leq (1 + \delta) \| x \|_2^2.
\end{align}
Thus, using that $\delta \leq \varepsilon/4,$

\begin{align}
\nonumber
\Big( 1 - \frac{\varepsilon}{4} \Big) \| x \|_2^2 \leq \sum_{J=1}^R\| \Phi_{(J)} D_{x_{(J)}} \xi_{(J)}  \|_2^2 \leq \Big( 1 + \frac{\varepsilon}{4} \Big) \| x \|_2^2.
\end{align}

\item
To estimate the second term, fix $\xi_{(1)} =: b$ and consider the random variable
\[
X = b^*D_{x_{(1)}}\Phi_{(1)}^*\Phi_{(\flat)} D_{x_{(\flat)}} \xi_{(\flat)} = \scalprod{v}{\xi_{(\flat)}}
\]
 with $v$ as in Proposition \ref{C}.   By Hoeffding's inequality (Proposition~\ref{hoeff}) combined with Proposition~\ref{C},
\begin{align}
\mathbb{P}(| X | \geq \gamma \varepsilon) \leq 2\exp{\Big(-\frac{s \gamma^2 \varepsilon^2}{2\delta^2} \Big)}.
\label{eq:conck2}
\end{align}

Taking a union bound, one obtains:
\begin{align}
\nonumber
\P\left( \exists x\in E: |X| \geq \gamma \varepsilon \right) \leq \exp\left(\log p + \log 2 -\frac{\gamma^2 s \varepsilon^2}{2\delta^2} \right).
\end{align}

In order for this probability to be less than $\eta/2$, we need:
\begin{align}
\nonumber
\log 2p -\frac{ s \gamma^2 \varepsilon^2}{2\delta^2} \leq \log \frac{\eta}{2},
\end{align}
that is,
\begin{align}
\label{cond1}
\delta \leq \varepsilon/4 \sqrt{\frac{ 8\gamma^2 s}{\log{(4p/\eta)}}}.
\end{align}
\item

We can rewrite the third term as
\begin{align}
\nonumber
 \sum_{\substack{J, L =2\\ J\neq L }}^R\scalprod{\Phi_{(J)} D_{x_{(J)}} \xi_{(J)}}{\Phi_{(L)} D_{x_{(L)}} \xi_{(L)}} =  \scalprod{\xi}{C \xi} = \sum_{j,\ell=s+1}^N \xi_j \xi_\ell C_{j\ell},
\end{align}
where $C \in \mathbb{R}^{N \times N}$ is the matrix as in Proposition \ref{C}.
By Proposition \ref{C}, we have $\|C\|\leq  \frac{\delta}{s}$ and $\|C\|_{\cal F} \leq  \frac{\delta}{\sqrt{s}}$ , hence by
Proposition~\ref{chaos}
\begin{align}
\P\left( \left| \sum_{j,\ell=s + 1}^N \xi_j \xi_\ell C_{j\ell} \right| \geq \tau \varepsilon \right)\leq 2 \exp\left(-\frac{1}{64}\min\left(\frac{ s \tau^2 \varepsilon^2}{\delta^2 }, \frac{96 \tau s\varepsilon}{ 65\delta} \right)\right). \label{eq:concm}
\end{align}
Using a union bound, one obtains:
\begin{align}
\nonumber
\P\left( \exists x\in E: \left| \sum_{j,\ell=s+1}^N \xi_j \xi_\ell C_{jl} \right| \geq \tau \varepsilon \right)\leq 2 \exp\left(\log p - \frac{1}{64} s \min\left(\frac{ \tau^2 \varepsilon^2}{\delta^2 }, \frac{96 \tau \varepsilon}{65 \delta} \right)\right).
\end{align}

In order for this probability to be less than $\eta/2$, we need:
\begin{align}
\nonumber
\log 2p-\frac{1}{64} s \min\left(\frac{ \tau^2\varepsilon^2}{\delta^2}, \frac{96\tau\varepsilon}{ 65\delta} \right)\leq \log{(\eta/2)} ,
\end{align}
that is,
\begin{align}
\label{cond2}
\delta \leq \frac{\varepsilon}{4} \min\left(\sqrt{ \frac{\tau^2 s}{4\log \left(\frac{4p}{\eta}\right)}},  \frac{\frac{96}{65}\tau s}{16 \log \left(\frac{4p}{\eta}\right)}\right).
\end{align}
\end{enumerate}
By assumption, $\delta \leq \frac{\varepsilon}{4}$, so
conditions \eqref{cond1} and \eqref{cond2} are satisfied by setting $\tau = .55, \gamma = .1$, and $s \geq 20 \log{(4p/\eta)}$ (that is, $k=2s\geq 40 \log{(4p/\eta)}$).  Then the second term is bounded by $.2\delta$ in absolute value, and the last term is bounded by $.55\delta$. Together with the deterministic RIP-based estimate for the first term, this implies the Theorem. \qed

\paragraph{Proof of Proposition~\ref{prop:conv}.}
Fix $\varepsilon > 0$, and suppose that there is a constant $c_3$ such that for all pairs $(k,m)$ with $k \leq c_3 \delta^2 m/\log(N/k)$, $\Phi=\Phi(m) \in \mathbb{R}^{m \times N}$ has the Restricted Isometry Property of order $k$ and level $\delta = \frac{\varepsilon}{4}$. Now let $(k,m)$ be admissible. An elementary monotonicity argument shows that there exists $k'\geq k$ such that $(k',m)$ is admissible and $k' \geq \frac{1}{2} c_3 \delta^2 m/\log(N/k')$. 
 Fix $x\in\R^N$ and let $\xi \in \mathbb{R}^N$ be a Rademacher sequence. Then, for any fixed vector $x \in \mathbb{R}^N$, the estimates in equations \eqref{eq:conck2} and \eqref{eq:concm} with parameters $\tau = .55$ and $\gamma=.1$ imply the existence of a constant $c_5 < 1$ for which
\begin{eqnarray}
\P\left(\Big| \| \Phi D_{\xi}x \|_2^2 - \| x \|_2^2 \Big| \geq \varepsilon \| x \|_2^2 \right)) &\leq& 2\exp{(-c_5 k')} \nonumber \\
&\leq & 2\exp{(-c_4 \varepsilon^2 m \log^{-1}(N/{k'}))}
\end{eqnarray}
where $c_4 = c_5 c_3/32$.   \qed

\paragraph{Remarks:}
Although we have stated the main result for the setting $x \in \mathbb{R}^N$ and $\Phi \in \mathbb{R}^{m \times N}$, all of the analysis holds also in the complex setting, $x \in \mathbb{C}^N$ and $\Phi \in \mathbb{C}^{m \times N}$.

As shown in \cite{badadewa08}, a random matrix $\Phi$ whose entries follow a subgaussian  distribution is known to have with high probability the Restricted Isometry Property of best possible order, that is, one can choose $m\asymp \delta^{-2} k \log\left(\frac{N}{k}\right).$ When $k\geq 40 \log\left(\frac{4p}{\eta}\right)$, $\Phi$ is a JL embedding by Theorem~\ref{thm:main}, and our resulting bound for $m$ is optimal up to a single logarithmic factor in $N$. This shows that Theorem~\ref{thm:main} must also be optimal up to a single logarithmic factor in $N$.

\subsection*{Acknowledgments}
The authors would like to thank Holger Rauhut, Deanna Needell, Jan Vyb\'iral, Mark Tygert, Mark Iwen, Justin Romberg, Mark Davenport, and Arie Israel for valuable discussions on this topic.  Rachel Ward gratefully acknowledges the partial support of National Science Foundation Postdoctoral Research Fellowship.  Felix Krahmer gratefully acknowledges the partial support of the Hausdorff Center for Mathematics.
Finally, both authors are grateful for the support of the Institute of Advanced Study through the Park City Math Institute where this project was initiated.
\bibliography{RIP_JL}
\bibliographystyle{abbrv}

\end{document}